%% file: arxiv_preprint_1.tex
\documentclass{llncs}

\usepackage{amssymb,amsmath}
\usepackage{graphicx,graphics}
\usepackage{algorithm,algorithmic,color}
\usepackage{subfigure}
\newcommand{\comment}[1]{}

\newcommand{\wt}{w}
\newcommand{\fwt}{F_\wt}
\newcommand{\complexity}[1]{O(|X||\pc|^2 + #1^4|\ms{\pig{\pc}}|^2)}

\newcommand{\ms}[1]{\Delta_{#1}}
\newcommand{\minms}[2]{\Delta_{#1}^{\min}}
\newcommand{\pmc}[1]{\Pi_{#1}}
\newcommand{\pf}{\mathrm{pf}}

\newcommand{\ind}[1]{I_{#1}}

\newcommand{\pc}{\mathcal{C}}
\newcommand{\pig}[1]{\mathrm{int}(#1)}

\newcommand{\wfill}{\mathrm{fill}_{\fwt}}

\newcommand{\wmfi}{\mathrm{mfi}_{\fwt}}

\newcommand{\xt}[1]{\mathcal{T}{#1}}
\newcommand{\xtree}[1]{\xt{#1} = (T{#1},\phi{#1})}

\newtheorem{observation}{Observation}

\begin{document}

\title{Potential Maximal Clique Algorithms for Perfect Phylogeny Problems}
\author{Rob Gysel}
\institute{Department of Computer Science, University of California, Davis, 1 Shields Avenue, Davis CA 95616, USA\\ \email{rsgysel@ucdavis.edu}}

\maketitle

\begin{abstract}
Kloks, Kratsch, and Spinrad showed how treewidth and min\-imum-fill, NP-hard combinatorial optimization problems related to minimal triangulations, are broken into subproblems by block subgraphs defined by minimal separators.
These ideas were expanded on by Bouchitt{\'e} and Todinca, who used potential maximal cliques to solve these problems using a dynamic programming approach in time polynomial in the number of minimal separators of a graph.
It is known that solutions to the perfect phylogeny problem, maximum compatibility problem, and unique perfect phylogeny problem are characterized by minimal triangulations of the partition intersection graph.
In this paper, we show that techniques similar to those proposed by Bouchitt{\'e} and Todinca can be used to solve the perfect phylogeny problem with missing data, the two-state maximum compatibility problem with missing data, and the unique perfect phylogeny problem with missing data in time polynomial in the number of minimal separators of the partition intersection graph.
\end{abstract}

\section{Introduction}
The perfect phylogeny problem, also called the character compatibility problem, is a classic NP-hard \cite{BFW92,S92} problem in phylogenetics \cite{FB01,SS_Book03}.
Characters that have a perfect phylogeny are called homoplasy-free, i.e.\ they map to a tree with no horizontal evolutionary events such as recombination or gene transfer.
For a collection of partially labeled (a.k.a.\ missing data) unrooted trees, one can construct characters that have a perfect phylogeny precisely when the collection has a compatible supertree \cite{SS_Book03}.
The more general problem of supertree estimation is of wide interest.

Solutions to the perfect phylogeny problem are characterized by the existence of restricted (minimal) triangulations of the partition intersection graph \cite{B74,M83,S92}, and minimal triangulations of the partition intersection graph also play an important role in two variants of this problem. 
The first, the maximum compatibility problem, asks to find the largest subset of a set of given characters that has a perfect phylogeny \cite{BHS05,GG11}, and the second, asks if a set of characters has a unique perfect phylogeny\footnote{When a set of characters $\pc$ has a unique perfect phylogeny, it is also common in the literature to say that $\pc$ \emph{defines} an $X-$tree.} \cite{SS02,GH07}.
Interestingly, the unique perfect phylogeny problem is NP-hard even when a perfect phylogeny for the characters is given \cite{BLS12,HS13}.
Despite considerable advances in the field of minimal triangulations, to our knowledge these results have not been extended to the aforementioned problems, although the use of such methods to solve at least the perfect phylogeny problem may have been alluded to (see p.2 of \cite{FKTV08}).

Bouchitt{\'e} and Todinca \cite{BT02} used potential maximal cliques to create the first algorithm that solves minimum-fill and treewidth in time polynomial in $|\ms{G}|$, and this algorithm was improved upon in \cite{FKTV08}.
In this paper, we show how to extend the potential maximal clique approach to solve the perfect phylogeny problem, the maximum compatibility problem, and the unique perfect phylogeny problem.
This approach is motivated by the following: first, the algorithms in \cite{BT01,FKTV08} run in time polynomial in the number of minimal separators of the graph, and second, that data generated by the coalescent-based program \emph{ms} \cite{H02} often results in a partition intersection graph with a reasonable number of minimal separators \cite{G10}, despite there being an exponential number of minimal separators in general.
In order to unify our approach, we use a weighted variant of the well-studied minimum-fill problem, which is NP-hard \cite{Y81} and is an active area of research \cite{BHV11,FKTV08}.

Given full characters (a.k.a.\ \emph{complete data}), the perfect phylogeny problem is solvable in polynomial time when the number of characters is fixed \cite{MWW94} or when the number of parts is bounded \cite{AF94}.
Our results apply to the most general setting, where the characters may be partial (a.k.a.\ \emph{missing data}), and each character has unbounded parts (a.k.a.\ \emph{unbounded maxstates}).
See \cite{H06} for a survey on minimal triangulations, \cite{FB01,SS_Book03} for further reading on the perfect phylogeny / character compatibility problem, and \cite{GH07} for further reading on unique perfect phylogeny.

\section{Definitions and results}
\begin{figure}[t]
\centering
\subfigure[$\xt{}$]{
   \scalebox{.6}{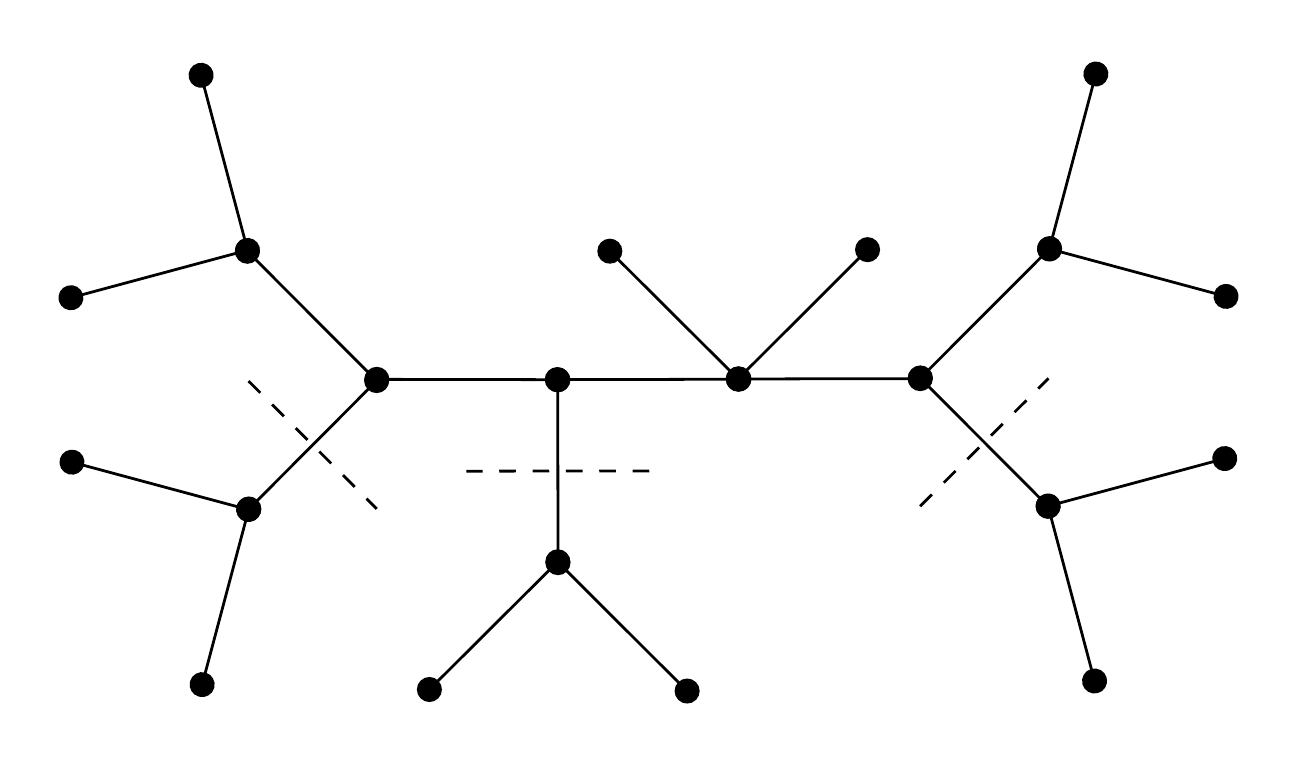}
 }
\subfigure[$\pig{\pc}$]{
   \scalebox{.6}{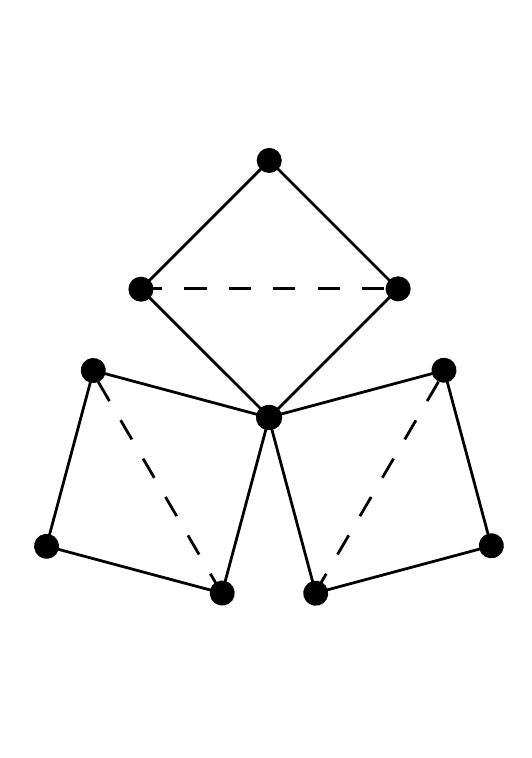}
 }
\caption{
An $X-$tree $\xt{}$ displaying $\pc = \{abcdef|gh|ij|kl, ag|dj|fl, bh|ci|ek\}$ and the corresponding partition intersection graph $\pig{\pc}$.
We use $A$ to denote $abcdef$, and have labeled the vertices by their cells.
The character $abcdef|gh|ij|kl$ distinguishes the edges of $\xt{}$ marked with dashes.
Removing these edges results in the four subtrees defined by $\xt{}(abcdef)$, $\xt{}(gh)$, $\xt{}(ij)$, and $\xt{}(kl)$.
The dashed edges of $\pig{\pc}$ define a proper triangulation, and the solid edges are obtained by cell intersection. 
If we replace $ag|dj|fl$ and $bh|ci|ek$ with the characters $ag|bh$, $ci|dj$, and $ek|fl$, we would obtain a partition intersection graph isomorphic to $\pig{\pc}$ but with a different coloring.
In that case, there is no proper triangulation because each four cycle has only two colors, and the fill edge $ag,bh$ is monochromatic.
Note that $\xt{}$ does not display $ag|bh$, $ci|dj$, or $ek|fl$.
}\label{fig_pp}
\end{figure}
An \emph{$X-$tree} is a pair $\xtree{}$, where $T$ is an undirected tree, and $\phi$ is a mapping from $X$ to the nodes of $T$ such that every node of $T$ with degree two or one is mapped to by $\phi$.
A \emph{character} on $X$ is a partition $\chi = A_1 | A_2 | \ldots | A_r$ of a subset of $X$.
For $i = 1, 2, \ldots, r$ the set $A_i$ is a \emph{cell} of $\chi$.
Given a cell $A$ of a character, the minimal subtree of $T$ that connects $\phi(A)$ is denoted $\xt{}(A)$.
An $X-$tree $\xt{}$ \emph{displays} a character $\chi$ if, for each pair of distinct cells $A$ and $A'$ of $\chi$, the trees $\xt{}(A)$ and $\xt{}(A')$ have no nodes in common.
Given a set $\pc$ of characters, the \emph{perfect phylogeny problem} is to determine if there is an $X-$tree $\xt{}$ that displays every character in $\pc$.
In this case, we call $\xt{}$ a \emph{perfect phylogeny} for $\pc$, and say that $\pc$ is \emph{compatible}.

The perfect phylogeny problem reduces to a graph theoretic problem that we detail now.
A graph is \emph{chordal} if any cycle it has on four or more vertices has a \emph{chord}, that is, an edge between two non-consecutive vertices in the cycle.
When $G$ is not chordal, we may add edges to $G$ to obtain a chordal supergraph $H$ that is called a \emph{triangulation} of $G$.
The edges added to $G$ to obtain $H$ are called \emph{fill edges} of $H$.
When no proper subset of $H$'s fill edges can be added to $G$ to obtain a triangulation, we call $H$ a \emph{minimal triangulation} of $G$.

Given a set of characters $\pc$, the \emph{partition intersection graph} $\pig{\pc}$ is the graph with vertex set $\{(A,\chi) \mid \chi \in \pc \mbox{ and } A \mbox{ is a cell of } \chi \}$, and two vertices $(A,\chi)$ and $(A',\chi')$ are adjacent in $\pig{\pc}$ if and only if $A$ and $A'$ have non-empty intersection.
If $A_1$ and $A_2$ are cells of a character $\chi$, then $A_1$ and $A_2$ are disjoint because $\chi$ is a partition of a subset of $X$, so $(A_1,\chi)$ and $(A_2,\chi)$ are not adjacent in $\pig{\pc}$.
The vertex $(A,\chi)$ has \emph{cell} $A$ and \emph{character} $\chi$.
A triangulation of $\pig{\pc}$ is \emph{proper} if, for each fill edge, the vertices involved in the fill edge have different characters.
This may be viewed as coloring each vertex $(A,\chi)$ of $\pig{\pc}$ by its character $\chi$, resulting in a properly colored graph, and then proper triangulations are those whose fill edges preserve the proper coloring.
If $u$ and $v$ are vertices of $\pig{\pc}$ that have the same character/color, we say that $u$ and $v$ are \emph{monochromatic}.
If a triangulation of $\pig{\pc}$ has $uv$ as an edge, we say that $uv$ is a \emph{monochromatic fill edge} of the triangulation.
See Figure \ref{fig_pp} for an example of these concepts.

For the remainder of this section, we characterize solutions to perfect phylogeny problems as constrained minimal triangulations of the partition intersection graph, and state our algorithmic results.
These problems will then be discussed in terms of minimum-weight minimal triangulations in Section 2, and we prove our computational results in Section 3, all of which rely on Algorithm \ref{alg_this_is_the_only_algorithm_in_the_paper}.
The connection between triangulations and perfect phylogeny stems from the following result.

\begin{theorem}\cite{B74,M83,S92}\label{thm_pig_pp}
Let $\pc$ be a set of characters on $X$.
Then $\pc$ is compatible if and only if $\pig{\pc}$ has a proper minimal triangulation.
\end{theorem}

While Theorem \ref{thm_pig_pp} was not originally stated in terms of minimal triangulations, it follows from the definitions that there is a proper triangulation if and only if there is a proper minimal triangulation.
The set of minimal separators of $\pig{\pc}$ are denoted $\ms{\pig{\pc}}$ (their definition appears in Section 3).
Our first algorithmic result is the following theorem.

\begin{theorem}\label{result_pp}
Let $\pc$ be a set of characters on $X$ with at most $r$ parts per character.
There is an $\complexity{(r|\pc|)}$ time algorithm that solves the perfect phylogeny problem.
\end{theorem}

If $\pc$ is not compatible, then the \emph{maximum compatibility problem} is to determine the largest subset $\pc^*$ of $\pc$ that is compatible, and $\pc^*$ is an \emph{optimal solution}.
In order to characterize solutions to the maximum compatibility problem in terms of minimal triangulations, we must consider non-proper triangulations of the partition intersection graph.
We say $\chi$ is \emph{broken by} a fill edge $(A,\chi)(A',\chi)$ because $\chi$ is the shared character of both vertices.
Given a triangulation $H$ of $\pig{\pc}$, the \emph{displayed characters} of $H$ are the characters of $\pc$ that are not broken by any fill edge.

\begin{theorem}\cite{BHS05,GG11}\label{thm_maxcompat}
Let $\pc$ be a set of characters on $X$.
Then $\pc^*$ is an optimal solution to the maximum compatibility problem if and only if there is a minimal triangulation $H^*$ of $\pig{\pc}$ that has $\pc^*$ as its displayed characters, and for every other minimal triangulation $H'$ of $\pig{\pc}$ with displayed characters $\pc'$, $|\pc'| \leq |\pc^*|$.
\end{theorem}

Given a set of characters $\pc$, a \emph{character weight} is a function $\wt$ from $\pc$ to the positive real numbers (i.e.\ excluding zero).
For a subset $\pc'$ of $\pc$, define $\wt(\pc') = \sum_{\chi \in \pc'} \wt(\chi)$.
The \emph{$\wt-$maximum compatibility problem} is to find a the subset $\pc^*$ of $\pc$ such that $\wt(\pc^*) = \max \wt(\pc')$, where the maximum is taken over all compatible subsets $\pc'$ of $\pc$.
We generalize Theorem \ref{thm_maxcompat} below, and reserve its proof for Section 2.

\begin{theorem}\label{thm_weightedmaxcompat}
Let $\pc$ be a set of characters on $X$ with character weight $\wt$.
Then $\pc^*$ is an optimal solution to the $\wt-$maximum compatibility problem if and only if there is a minimal triangulation $H^*$ of $\pig{\pc}$ that has $\pc^*$ as its displayed characters, and for any other minimal triangulation $H'$ of $\pig{\pc}$ with displayed characters $\pc'$, $\wt(\pc') \leq \wt(\pc^*)$.
\end{theorem}

Our second algorithmic result is for two-state characters only.
Such characters are interesting because they are related to finding compatible supertrees.
In that context, an optimal solution to maximum compatibility corresponds to a supertree that agrees with the most edges from the partially labeled trees given as input.

\begin{theorem}\label{result_binary_cr}
Let $\pc$ be a set of ($\wt-$weighted) two-state characters on $X$, i.e., each $\chi \in \pc$ has two cells.
There is an $\complexity{|\pc|}$ time algorithm that solves the ($\wt$-)maximum compatibility problem.
\end{theorem}

The \emph{unique perfect phylogeny problem} is to determine if a perfect phylogeny for a set of characters is the only perfect phylogeny for those characters.
An edge $uv$ of an $X-$tree $\xt{}$ is \emph{distinguished} by a character $\chi$ if contracting $uv$ results in an $X-$tree that does not display $\chi$, and $\xt{}$ is \emph{distinguished} by $\pc$ if each edge of $\xt{}$ is distinguished by a character of $\pc$.
An $X-$tree $\xtree{}$ is \emph{ternary} if every internal node of $T$ has degree three.
Semple and Steel characterized the existence of a unique perfect phylogeny as follows.

\begin{theorem}\label{thm:unique_pp}\cite{SS02}
Let $\pc$ be a set of characters on $X$.
Then $\pc$ has a unique perfect phylogeny $\xtree{}$ if and only if the following conditions hold:
\begin{description}
	\item[1.] there is a ternary perfect phylogeny $\xtree{}$ for $\pc$ and $\xt{}$ is distinguished by $\pc$;
	\item[2.] $\pig{\pc}$ has a unique proper minimal triangulation.
\end{description}
\end{theorem}

It is well known how to create a perfect phylogeny $\xtree{}$ for $\pc$ from a \emph{clique tree} of a proper minimal triangulation in polynomial time (e.g.\ see the proof of Lemma 5.1 in \cite{BHS05}), and a clique tree of a chordal graph can be computed in linear time \cite{BP92}.
Checking if $\xt{}$ is ternary and distinguished by $\pc$ is also easy to do: an edge $uv$ is distinguished by $\chi$ if and only if $u$ is a node of $\xt{}(A)$ and $v$ is a node of $\xt{}(A)$ for distinct cells $A$, $A'$ of $\chi$.
So if it is known that $\pig{\pc}$ has a unique proper minimal triangulation, it is possible to determine if $\pc$ has a unique perfect phylogeny in polynomial time.
On the other hand, it has recently been shown \cite{BLS12,HS13} that if a perfect phylogeny is given for a set of characters, it is still NP-hard to determine if it is the unique perfect phylogeny for those characters\footnote{These papers show that this problem is NP-hard even when the characters are \emph{quartet trees}, which in our setting correspond to characters of the form $ab | cd$.}.
That is, determining if $\pig{\pc}$ has a unique proper minimal triangulation is NP-hard \cite{HS13}.
This makes our last algorithmic result of interest.

\begin{theorem}\label{result_upp}
Let $\pc$ be a set of characters on $X$ with at most $r$ parts per character.
There is an $\complexity{(r|\pc|)}$ time algorithm that determines if $\pig{\pc}$ has a unique proper minimal triangulation, i.e. it solves the unique perfect phylogeny problem.
\end{theorem}

\section{Characterizations via weighted minimum-fill}

In this section, we characterize solutions to the perfect phylogeny problem and maximum compatibility problem as a weighted-variant of the \emph{minimum-fill problem}, which asks for the fewest number of edges required to triangulate a graph.
A similar characterization will be given for solutions to the unique perfect phylogeny problem, that has an additional requirement on the minimal separators involved in zero-weight minimal triangulations.
In order for our results to be useful in the next section, each result will be given with respect to minimal triangulations.

Suppose $G$ is a non-complete graph.
If $U$ is a subset of $G$'s vertices, then the \emph{potential fill edges} $\pf(U)$ of $U$ are pairs of vertices of $U$ that are not edges of $G$.
A \emph{fill weight} on $G = (V,E)$ is a function $\fwt$ from $\pf(V)$ to the non-negative real numbers, i.e., including zero.
For a triangulation $H$ of $G$ with fill weight $\fwt$, the \emph{weight} of $H$ is $\fwt(H) = \sum \fwt(f)$ where the sum occurs over all fill edges of $H$.
We will call $H$ a \emph{$\fwt-$minimum triangulation} of $G$ if, for every other triangulation $H'$ of $G$, $\fwt(H) \leq \fwt(H')$.
In this case we write $\wmfi(G) = \fwt(H)$.
If $\fwt(H) = 0$, then $H$ is a \emph{$\fwt-$zero triangulation} of $G$.
If $H$ is a $\fwt-$minimum or $\fwt-$zero triangulation that is also a minimal triangulation of $G$, then $H$ is a \emph{$\fwt-$min\-imum minimal triangulation} or \emph{$\fwt-$zero minimal triangulation}, respectively.
Note that if a $\fwt-$zero triangulation exists, it must be a $\fwt-$minimum triangulation.
Additionally, because $\fwt$ is non-negative, there is always a minimal triangulation that is a $\fwt-$minimum triangulation.

\begin{definition}
Let $\pc$ be a set of characters on $X$.
Then $\ind{\pc}$ is the fill weight of $\pig{\pc}$ defined by
\begin{equation*}
\ind{\pc}(uv) = \left\{ \begin{array}{ll}
         1 & \mbox{if $u$ and $v$ are monochromatic};\\
         0 & \mbox{otherwise}.\end{array} \right.
\end{equation*}
\end{definition}

\begin{observation}\label{obs_proper_zero}
Let $\pc$ be a set of characters on $X$.
Then a triangulation $H$ of $\pig{\pc}$ is proper if and only if $\ind{\pc}(H) = 0$.
\end{observation}

\begin{lemma}\label{lem_pp_reduction}
A collection $\pc$ of characters on $X$ are compatible if and only if $\pig{\pc}$ has a $\ind{\pc}$-zero minimal triangulation.
\end{lemma}
\begin{proof}
The lemma follows from Theorem \ref{thm_pig_pp} and Observation \ref{obs_proper_zero}.
\qed\end{proof}

The following two lemmas, which follow from results in \cite{BHS05,GG11}, will be helpful for proving Theorem \ref{thm_weightedmaxcompat}.

\begin{lemma}\label{lem:improper_tri}
Suppose $\pc$ is a set of characters and $\pc' \subseteq \pc$ is compatible.
Then there is a minimal triangulation of $\pig{\pc}$ and $\pc'$ is a subset of its displayed characters.
\end{lemma}

\begin{lemma}\label{lem:displayedchar_compatible}
Suppose $\pc$ is a set of characters and $H$ is a triangulation of $\pig{\pc}$.
Then the displayed characters of $H$ are a compatible subset of $\pc$.
\end{lemma}

\textit{(Proof of Theorem \ref{thm_weightedmaxcompat})}
Let $\pc^*$ be an optimal solution to the $\wt-$maximum compatibility problem.
By Lemma \ref{lem:improper_tri}, there is a minimal triangulation $H^*$ of $\pig{\pc}$ that has at least $\pc^*$ as its displayed characters.
Displayed character sets are compatible by Lemma \ref{lem:displayedchar_compatible}, so by positivity of $\wt$ and optimality of $\pc^*$, the displayed characters of $H^*$ are exactly $\pc^*$.
If $H'$ is another minimal triangulation of $\pig{\pc}$ with displayed character set $\pc(H')$, then $\pc(H')$ is compatible by Lemma \ref{lem:displayedchar_compatible}, so $\wt(\pc(H')) \leq \wt(\pc^*)$ by optimality of $\pc^*$.

For the converse, let $H$ be a minimal triangulation of $\pig{\pc}$ with displayed characters $\pc(H)$, and suppose $\wt(\pc(H))$ is greater than the weight of the displayed characters of any other minimal triangulation of $\pig{\pc}$.
Then $\wt(\pc^*) \leq \wt(\pc(H))$ because $\pc^*$ are the displayed characters of $H^*$.
By Lemma \ref{lem:displayedchar_compatible} the set $\pc(H)$ is compatible, so $\wt(\pc^*) = \wt(\pc(H))$ by optimality of $\pc^*$.
Therefore $\pc(H)$ is an optimal solution.
\qed

\begin{definition}
Let $\pc$ be a set of characters on $X$ that are weighted by $\wt$.
Then the fill weight $\fwt$ of $\pig{\pc}$ induced by $\wt$ is
\begin{equation*}
\fwt(uv) = \left\{ \begin{array}{ll}
         \wt(\chi) & \mbox{if $u$ are $v$ monochromatic and colored by $\chi$};\\
         0 & \mbox{otherwise}.\end{array} \right.
\end{equation*}
\end{definition}

\begin{lemma}\label{lem_fill_sum}
Let $\pc$ be a collection of two-state characters weighted by $\wt$, and suppose $H$ is a triangulation of $\pig{\pc}$ with displayed characters $\pc(H)$.
Then $\wt(\pc) = \fwt(H) + \wt(\pc(H))$.
\end{lemma}
\begin{proof}
For each $\chi$ in $\pc$ there is exactly one potential fill edge $uv$ of $\pig{\pc}$ such that $u$ and $v$ are monochromatic with shared character $\chi$ because $\chi$ has two states.
In particular, if $\chi = A | A'$ then $u = (A,\chi)$ and $v = (A',\chi)$.
Hence there is a one-to-one correspondence between characters in $\pc$ and potential fill edges stemming from monochromatic pairs of vertices of $\pig{\pc}$.
Further, each monochromatic pair of vertices is either a fill edge of $H$, or it corresponds to a displayed character of $H$.
Any other potential fill edge $u'v'$ of $\pig{\pc}$ that does not arise in this way is not monochromatic, and in this case $\fwt(u'v') = 0$.
Letting $X$ be the set of monochromatic potential fill edges of $\pig{\pc}$, we have
\begin{align*}
	\wt(\pc) &= \sum_{f \in X} \fwt(f) \\
			&= \sum_{f \in X \cap E(H)} \fwt(f) + \sum_{f \in X - E(H)}  \fwt(f) \\
			&= \fwt(H) + \wt(\pc(H)) \enspace .
\end{align*}
\qed\end{proof}

\begin{theorem}\label{thm_binary_cr_reduction}
Let $\pc$ be a collection of two-state characters weighted by $\wt$.
Then $\pc^*$ is a $\wt-$maximum compatible subset of $\pc$ if and only if there is a $\fwt$-minimum minimal triangulation $H^*$ of $\pig{\pc}$ that has $\pc^*$ as its displayed characters.
\end{theorem}
\begin{proof}
Suppose that $\pc^*$ is a $\wt-$maximum compatible subset of $\pc$.
By Theorem \ref{thm_weightedmaxcompat}, there is a minimal triangulation $H^*$ of $\pig{\pc}$ that has $\pc^*$ as its displayed characters.
For the sake of contradiction suppose $H^*$ is not a $\fwt-$minimum minimal triangulation, so there is a triangulation $H$ of $\pig{\pc}$ such that $\fwt(H) < \fwt(H^*)$.
Letting $\pc(H)$ be the displayed characters of $H$, by Lemma \ref{lem_fill_sum} we have $\wt(\pc) - \wt(\pc(H)) < \wt(\pc) - \wt(\pc^*)$ and therefore $\wt(\pc^*) < \wt(\pc(H))$.
This contradicts the optimality of $\pc^*$, so $H^*$ must be a $\fwt-$minimum minimal triangulation.

Now let $H'$ be a $\fwt-$minimum minimal triangulation of $\pig{\pc}$ with displayed characters $\pc(H')$.
Then $\fwt(H') = \fwt(H^*)$ by $\fwt-$minimization, and $\wt(\pc) - \wt(\pc(H')) = \wt(\pc) - \wt(\pc^*)$ by Lemma \ref{lem_fill_sum} so $\wt(\pc(H')) = \wt(\pc^*)$.
The set $\pc(H)$ is compatible by Lemma \ref{lem:displayedchar_compatible}, so $\pc(H)$ is an optimal solution.
\qed\end{proof}

The weighted maximum compatibility problem can be used to solve the maximum compatibility problem by using the character weight where each $\chi \in \pc$ has weight one.
This character weighting induces the fill weight $\ind{\pc}$, giving the following corollary.

\begin{corollary}\label{cor_binary_cr_reduction}
Let $\pc$ be a collection of two-state characters.
Then $\pc^*$ is a maximum compatible subset of $\pc$ if and only if there is a $\ind{\pc}$-minimum minimal triangulation $H^*$ of $\pig{\pc}$ that has $\pc^*$ as its displayed characters.
\end{corollary}

We conclude this section by characterizing solutions to unique perfect phylogeny.
Let $G = (V,E)$ be an undirected graph and $S \subseteq V$.
We will use $G - S$ to denote the graph obtained from $G$ by removing the vertices $S$ and edges that are incident to a vertex in $S$.
If $x$, $y$ are connected vertices in $G$ but disconnected in $G - S$, then $S$ is an \emph{$xy-$separator}.
When no proper subset of $S$ is also an $xy-$separator, then $S$ is a \emph{minimal $xy-$separator}\footnote{Note that a minimal $xy-$separator $S$ is defined with respect to $x$ and $y$. That is, it may be that there is a different pair of vertices $u$,$v$ of $G$ such that $S$ is a non-minimal $uv-$separator.}.
If there is at least one pair of vertices $x$ and $y$ such that $S$ is a minimal $xy-$separator, then it is a \emph{minimal separator} of $G$.
The set of minimal separators of $G$ is denoted by $\ms{G}$.
Suppose $\Phi$ is a subset of $G$'s minimal separators.
The graph $G_\Phi$ is obtained from $G$ by adding the fill edge $uv$ whenever $uv \in \pf(S)$ for some $S$ in $\Phi$, and we say $G$ is obtained by \emph{saturating} each minimal separator in $\Phi$.
The following fundamental result characterizes the minimal triangulations of a graph in terms of its minimal separators.

\begin{theorem}\cite{PS95,PS97} see also \cite{KKS97}\label{thm_mtt}
Let $G$ a graph and $\ms{G}$ its minimal separators.
If $H$ is a minimal triangulation of $G$, then $\ms{H}$ is a maximal pairwise-parallel set of minimal separators of $G$ and $H = G_{\ms{H}}$.
Conversely, if $\Phi$ is any maximal pairwise-parallel set of minimal separators of $G$, then $G_\Phi$ is a minimal triangulation of $G$ and $\ms{G_\Phi} = \Phi$.
\end{theorem}

An important observation from this theorem is that if $H$ is a minimal triangulation of $G$, then $\ms{H} \subseteq \ms{G}$.
Let $\pc$ be a set of characters and $\fwt$ be a fill weight on $\pig{\pc}$.
We will use $\minms{\fwt}{\pig{\pc}}$ to denote the set of minimal separators $S$ of $\pig{\pc}$ such that there is a $\fwt-$minimum minimal triangulation $H$ of $\pig{\pc}$ with $S \in \ms{H}$.

\begin{theorem}\label{thm_upp_reduction}
Suppose $\pc$ is a collection of characters on $X$.
Then $\pig{\pc}$ has a unique proper minimal triangulation if and only if 
\begin{enumerate}
	\item $\pig{\pc}$ has a $\ind{\pc}-$zero minimal triangulation; and
	\item $\minms{\ind{\pc}}{\pig{\pc}}$ is a maximal set of pairwise-parallel minimal separators of $\pig{\pc}$.
\end{enumerate}
\end{theorem}
\begin{proof}
Suppose $\pig{\pc}$ has a unique proper minimal triangulation $H^*$.
By Observation \ref{obs_proper_zero} it is a $\ind{\pc}-$zero minimal triangulation of $\pig{\pc}$, and each minimal separator of $H^*$ is a minimal separator of $\pig{\pc}$ by Theorem \ref{thm_mtt}, so $\ms{H^*} \subseteq \minms{\ind{\pc}}{\pig{\pc}}$.
Alternatively, if $S \in \minms{\ind{\pc}}{\pig{\pc}}$, then $S$ is a minimal separator of a $\ind{\pc}-$zero minimal triangulation of $\pig{\pc}$.
This minimal triangulation is proper by Observation \ref{obs_proper_zero}, so $S \in \ms{H^*}$ by uniqueness.
Therefore $\minms{\ind{\pc}}{\pig{\pc}} = \ms{H^*}$, and $\minms{\ind{\pc}}{\pig{\pc}}$ is a maximal pairwise-parallel set of minimal separators of $\pig{\pc}$ by Theorem \ref{thm_mtt}.

To prove the converse, suppose that $\pig{\pc}$ has a $\ind{\pc}-$zero minimal triangulation, and $\minms{\ind{\pc}}{\pig{\pc}}$ is a maximal set of pairwise-parallel minimal separators of $\pig{\pc}$.
By Theorem \ref{thm_mtt}, the graph $H$ obtained from $\pig{\pc}$ by saturating each minimal separator in $\minms{\ind{\pc}}{\pig{\pc}}$ is a minimal triangulation of $\pig{\pc}$, and further, for each fill edge $uv$ of $H$, there is a $S' \in \minms{\ind{\pc}}{\pig{\pc}}$ such that $u,v \in S'$.
By definition there is some $\ind{\pc}-$zero minimal triangulation that has $S'$ as a minimal separator.
This triangulation has $uv$ as a fill edge by Theorem \ref{thm_mtt} so $\ind{\pc}(uv) = 0$.
Therefore $H$ is an $\ind{\pc}-$zero minimal triangulation of $\pig{\pc}$, and by Observation \ref{obs_proper_zero}, $H$ is a proper minimal triangulation of $\pig{\pc}$.

Now let $H'$ be any proper minimal triangulation of $\pig{\pc}$.
By Observation \ref{obs_proper_zero}, $H'$ is an $\ind{\pc}-$zero minimal triangulation of $\pig{\pc}$, so $\ms{H'} \subseteq \minms{\ind{\pc}}{\pig{\pc}}$.
We assumed $\minms{\ind{\pc}}{\pig{\pc}}$ is pairwise-parallel, and $\ms{H'}$ is maximal with respect to being pairwise-parallel by Theorem \ref{thm_mtt}, so $\ms{H'} = \minms{\ind{\pc}}{\pig{\pc}}$.
Thus both $H$ and $H'$ are obtained from $\pig{\pc}$ by saturating each minimal separator of $\ms{H'} = \minms{\ind{\pc}}{\pig{\pc}}$, so $H' = H$.
Therefore $H$ is the unique proper minimal triangulation of $\pig{\pc}$.
\qed\end{proof}

\section{Finding weighted minimum triangulations}

In this section we show that, given a fill weight $\fwt$ for $G$, both $\wmfi(G)$ and $\minms{\fwt}{\pig{\pc}}$ can be computed in $\complexity{(r|\pc|)}$ time.
After that, we present proofs of our algorithmic results.

Given a graph $G$ and $X \subseteq V$, a set $C \subseteq V - X$ is a \emph{connected component} of $G - X$ if it is connected in $G - X$ and it is maximal with respect to this property.
A \emph{block} of a graph $G$ is a pair $(S,C)$ where $S \in \ms{G}$ and $C$ is a connected component of $G - S$, and it is \emph{full} or \emph{full with respect to $S$} if every vertex of $S$ has at least one neighboring vertex that is in $C$ (we write $N(C) = S$).
The \emph{realization} of a block $(S,C)$ is the graph $R(S,C)$ with vertex set $S \cup C$, and for any $u$ and $v$ in $S \cup C$, $uv$ is an edge of $R(S,C)$ if either $uv$ is an edge of $G$ or $uv \in \pf(S)$.

Kloks, Kratsch, and Spinrad \cite{KKS97} showed that the minimal triangulations of $G$ that have $S \in \ms{G}$ as a minimal separator (i.e.\ $S$ is saturated to obtain the minimal triangulation) can be obtained by independently minimally triangulating $R(S,C)$ for each connected component $C$ of $G - S$.
They used this fact to relate minimum fill to the realizations of the blocks of a minimal separator, an important first consideration for computing minimum fill using potential maximal cliques and minimal separators.
We extend this fact to weighted-minimum fill with the following lemma, whose proof follows with a slight modification of the proof of Theorem 3.4 in \cite{KKS97}, so we omit it.

\begin{lemma}\label{lem_decompose}
Let $G$ be a non-complete graph and $\fwt$ be a fill weight on $G$.
Then
\begin{equation*}
\wmfi(G) = \min_{S \in \ms{G}} (\wfill(S) + \sum_{C} \wmfi(R(S,C)))
\end{equation*}
where the sum occurs over the connected components $C$ of $G - S$ and
\begin{equation*}
\wfill(S) = \sum_{f \in \pf(S)} \fwt(f) \enspace .
\end{equation*}
\end{lemma}

It turns out that non-full blocks with respect to $S \in \ms{G}$ are full blocks with respect to a different minimal separator of $G$.
They also allow us to compute $\wmfi(R(S,C))$, which is a useful fact for later when we restrict our attention to full blocks of $G$.

\begin{lemma}\cite{BT01}\label{lem_nonfull_blocks}
Let $G$ be a graph, $S \in \ms{G}$, and $C$ be a connected component of $G - S$.
If $N(C) = S' \subset S$, then $(S',C)$ is a full block of $G$ (i.e.\ $S' \in \ms{G}$).
Further, if $E' \subseteq \pf(C)$, then the graph obtained from $R(S,C)$ by adding the fill edges in $E'$ is a minimal triangulation of $R(S,C)$ if and only if the graph obtained from $R(S',C)$ by adding the fill edges in $E'$ is a minimal triangulation of $R(S,C)$.
\end{lemma}

This gives us the following, an extension of Corollary 4.5 in \cite{BT01}.

\begin{corollary}\label{cor_nonfull_blocks}
Let $G$ be a graph, $S \in \ms{G}$, and $C$ be a connected component of $G - S$.
If $N(C) = S' \subset S$, then $\wmfi(R(S,C)) = \wmfi(R(S',C))$ for any fill weight $\fwt$.
\end{corollary}

In order to compute $\wmfi(R(S,C))$, we need the notion of a potential maximal clique.
Let $G$ be a graph and $K$ be a subset of its vertices.
Then $K$ is a \emph{potential maximal clique} of $G$ if there is a minimal triangulation $H$ of $G$ and $K$ is a maximal clique of $H$.
That is, every pair of vertices in $K$ are adjacent in $H$, and no proper superset of $K$ has this property.
The set of potential maximal cliques of $G$ is denoted by $\pmc{G}$.
The next two lemmas describe the interplay between potential maximal cliques, minimal separators, and blocks.
\begin{lemma}\cite{BT01}
Let $G$ be a graph and $K$ be a potential maximal clique of $G$.
Then $S \in \ms{G}$ and $S \subseteq K$ if and only if $N(C) = S$ for some connected component $C$ of $G - K$.
\end{lemma}

Therefore if $K \in \pmc{G}$ and $C_1, C_2, \ldots, C_k$ are the connected components of $G - K$, each $(S_i,C_i)$ where $N(C_i) = S_i$ is a full block of $G$ (i.e.\ $S_i \in \ms{G}$).
These blocks are called the blocks \emph{associated} to $K$.

\begin{lemma}\cite{BT01}\label{lem_mt_of_blocks}
Suppose $G$ is a graph, $S \in \ms{G}$, and $(S,C)$ is a full block.
Then $H(S,C)$ is a minimal triangulation of $R(S,C)$ if and only if 
\begin{enumerate}
	\item there is a potential maximal clique $K$ of $G$ such that $S \subset K \subseteq (S,C)$; and
	\item letting $(S_i,C_i)$ for $1 \leq i \leq p$ be the blocks associated to $K$ such that $S_i \cup C_i \subset S \cup C$, we have $E(H) = \bigcup_{i=1}^p E(H_i) \cup \pf(K)$ where $H_i$ is a minimal triangulation of $R(S_i,C_i)$ for each $1 \leq i \leq p$.
\end{enumerate}
\end{lemma}

The following lemma is an extension of Corollary 4.8 in \cite{BT01}.
For completeness, we provide a proof.

\begin{lemma}
Let (S,C) be a full block of $G$ and $\fwt$ be a fill weight on $G$.
Then
\begin{equation}\label{eqn}
\wmfi(R(S,C)) = \min_{S \subset K \subseteq (S,C)} (\wfill(K) - \wfill(S) + \sum \wmfi(R(S_i,C_i)))
\end{equation}
where the minimum is taken over all $K \in \pmc{G}$ such that $S \subset K \subseteq (S,C)$, and $(S_i,C_i)$ are the blocks associated to $K$ in $G$ such that $S_i \cup C_i \subset S \cup C$.
\end{lemma}
\begin{proof}
Let $H(S,C)$ be a triangulation of $R(S,C)$ such that $\wmfi(R(S,C)) = \fwt(H(S,C))$.
Without loss of generality, we may assume $H(S,C)$ is a minimal triangulation of $R(S,C)$ because $\fwt$ is non-negative.
By Lemma \ref{lem_mt_of_blocks}, there is a potential maximal clique $K$ such that $S \subset K \subseteq (S,C)$ with blocks $(S_i,C_i)$ associated to $K$ such that $S_i \cup C_i \subset S \cup C$ for $1 \leq i \leq p$.
Further, the fill edges of $H(S,C)$ are disjointly obtained from the fill edges of $H_i$ for $1 \leq i \leq p$ and $\pf(K) - \pf(S)$ (because $S$ is already saturated in $R(S,C)$).
Therefore $\wmfi(R(S,C)) = \fwt(H(S,C)) = \wfill(K) - \wfill(S) + \sum_{i=1}^p \fwt(H_i)$.

Now, for a given $1 \leq k \leq p$, suppose for the sake of contradiction that $H_k$ is not a $\fwt-$minimum fill of $R(S_k,C_k)$.
Then there is a minimal triangulation $H'_k$ of $R(S_k,C_k)$ such that $\fwt(H'_k) < \fwt(H_k)$.
Further, the graph $H'(S,C)$ with vertex set $(S,C)$ and edge set $E(H(S,C)) - E(H_k) \cup E(H'_k)$ is a minimal triangulation of $R(S,C)$ by Lemma \ref{lem_mt_of_blocks}, and it has a weighted fill of $\fwt(H'(S,C)) = \fwt(H(S,C)) - \fwt(H_k) + \fwt(H'_k) < \fwt(H(S,C))$.
This contradicts the $\fwt-$minimality of $H(S,C)$, so it must be that $\fwt(H_k) = \wmfi(R(S_k,C_k))$, and therefore $\wmfi(R(S,C)) = \wfill(K) - \wfill(S) + \sum_{i=1}^p \wmfi(R(S_i,C_i))$.
Letting LHS and RHS denote the left-hand side and right-hand side of equation (\ref{eqn}), respectively, we have shown that LHS $\geq$ RHS.

Now suppose $K^* \in \pmc{G}$ such that $S \subset K \subseteq (S,C)$ and 
\begin{equation*}
	\wfill(K^*) - \wfill(S) + \sum \wmfi(R(S^*_i,C^*_i)) = \mbox{RHS}
\end{equation*}
where $(S^*_i,C^*_i)$ are the blocks associated to $K^*$ in $R(S,C)$ for $1 \leq i \leq p^*$.
For $1 \leq i \leq p^*$, let $H^*_i$ be a minimal triangulation of $R(S^*_i,C^*_i)$ such that $\fwt(H^*_i) = \wmfi(R(S^*_i,C^*_i))$.
By Lemma \ref{lem_mt_of_blocks}, there is a minimal triangulation $H^*(S,C)$ of $R(S,C)$ obtained by adding the fill edges $\pf(K^*) - \pf(S)$ and $E(H^*_i) - E(R(S^*_i,C^*_i))$ for $1 \leq i \leq p^*$, and hence $\fwt(H^*(S,C)) = \wfill(K^*) - \wfill(S) + \sum \wmfi(R(S^*_i,C^*_i))$.
Now, $\wmfi(R(S,C)) \leq \fwt(H^*(S,C))$ by definition, so LHS $\leq$ RHS and therefore LHS $=$ RHS.
\qed\end{proof}

\begin{theorem}\label{thm_algorithm}
Let $\pc$ be a set of partial characters on $X$ with at most $r$ parts per character, and $\fwt$ be a fill weight on $\pig{\pc}$.
There is an $\complexity{(r|\pc|)}$ algorithm that computes $\wmfi(\pig{\pc})$ and $\minms{\fwt}{\pig{\pc}}$.
\end{theorem}

\begin{algorithm}\caption{}\label{alg_this_is_the_only_algorithm_in_the_paper}
\begin{algorithmic}[1]
	\STATE \textbf{Input:} Partial characters $\pc$ on $X$ with at most $r$ parts
	\STATE \textbf{Output:} $\wmfi(\pig{\pc})$ and $\minms{\fwt}{\pig{\pc}}$
	
	\STATE compute $\pig{\pc}$
	\STATE compute $\ms{\pig{\pc}}$ and $\pmc{\pig{\pc}}$
	
	\COMMENT{ Find the $\fwt-$minimum fill value for each full block }
	\STATE compute all the full blocks $(S,C)$ and sort them by the number of vertices
	
	\FOR{ each full block $(S,C)$ taken in increasing order }
		\STATE $\wmfi(R(S,C)) \gets \wfill(S \cup C)$ if $(S,C)$ is inclusion-minimal \\
			and $\wmfi(R(S,C)) \gets \infty$ otherwise
		\FOR{ each potential maximal clique $K$ s.t.\ $S \subset K \subseteq S \cup C$ }
			\STATE compute the blocks $(S_i,C_i)$ associated to $K$ s.t.\ $S_i \cup C_i \subset S \cup C$
			\STATE $\mbox{newfill} \gets \wfill(K) - \wfill(S) + \sum_{i} \wmfi(R(S_i,C_i))$
			\STATE $\wmfi(R(S,C)) \gets \min( \wmfi(R(S,C)), \mbox{newfill})$
		\ENDFOR
	\ENDFOR
	
	\STATE $\wmfi(\pig{\pc}) \gets \infty$
	
	\COMMENT{ Find the $\fwt-$minimum fill value for minimal triangulations containing $S$ }
	\FOR{ each minimal separator $S$ of $\pig{\pc}$ }
		\STATE compute the blocks $(S_i,C_i)$ associated to $S$ where $N(C_i) = S_i$
		\STATE $\wmfi(S) \gets \wfill(S) + \sum_{i} \wmfi(R(S_i,C_i))$
		\STATE $\wmfi(\pig{\pc} \gets \min( \wmfi(\pig{\pc}), \wmfi(S) )$
	\ENDFOR
	
	\STATE $\minms{\ind{\pc}}{\pig{\pc}} \gets \{ S \in \ms{\ind{\pc}} \mbox{ s.t. } \wmfi(S) = \wmfi(\pig{\pc})\}$

\end{algorithmic}
\end{algorithm}

\begin{proof}
Our approach is described in Algorithm \ref{alg_this_is_the_only_algorithm_in_the_paper}.
Constructing $\pig{\pc}$ can be done in $O((|X| + r^2)|\pc|^2)$ time as follows.
There are at most $r|\pc|$ vertices of $\pig{\pc}$, one per part of each character.
Recall that a pair of vertices $(A,\chi)$ and $(A',\chi')$ of $\pig{\pc}$ form an edge if and only if there is some $a \in A \cap A'$.
For each $a \in X$, let $\pc(a)$ be the vertices of $\pig{\pc}$ whose cell contains $a$.
These sets are computed in $O(|X||\pc|)$ amortized time by scanning each cell of each character.
The edges of $\pig{\pc}$ are now found by examining each pair $(A_1,\chi_1)(A_2,\chi_2)$ in $\pc(a)$ for all $a \in X$.
To address redundancy, order the characters and cells of each characters, then construct a table to check if $(A_1,\chi_1)$ and $(A_2,\chi_2)$ has already been found as an edge.
Examining $\pc(a)$ to find these edges takes $O(|\pc|^2)$ time (because $a$ is in at most one cell per character), and constructing the redundancy table takes $O((r|\pc|)^2)$ time, for a total of $O((|X| + r^2)|\pc|^2)$ time.

For a general graph $G$ with $|V|$ vertices and $|E|$ edges, it is possible to compute $\ms{G}$ in $O(|V|^3|\ms{G}|)$ time \cite{BBC00} and $\pmc{G}$ in $O(|V|^2|E||\ms{G}|^2)$ time \cite{BT02}.
Let $n$ be the number of vertices of $\pig{\pc}$.
The full block computation and nested for loop can be implemented in $O(n^3|\pmc{\pig{\pc}}|)$ time, which follows from the proof of Theorem 3.4\ in \cite{FKTV08}.
It is known that $|\pmc{G}| \leq |V||\ms{G}|^2 + |V||\ms{G}| + 1$ \cite{BT02}, so the nested for loop takes $O(n^4|\ms{\pig{\pc}}|^2)$ time.

Consider the second for loop and let $S \in \ms{\pig{\pc}}$.
The blocks associated to $S$ are found in $O(n^2)$ time by searching the graph to find the connected components, and then computing $N(C)$ for each connected component $C$ of $G - S$ (in the second computation, each edge of the graph is examined at most once).
By Lemma \ref{lem_nonfull_blocks}, each $(S_i,C_i)$ is a full block of $G$, so we have calculated $\wmfi(R(S_i,C_i))$ during the first for loop.
The calculation on line 17 matches the one in Lemma \ref{lem_decompose} because $\wmfi(R(S,C_i)) = \wmfi(R(S_i,C_i))$ by Corollary \ref{cor_nonfull_blocks}.
It takes $O(n^2)$ time to compute $\wfill(S)$, so the second for loop takes $O(n^2|\ms{\pig{\pc}}|)$ time.
The last line of the algorithm takes $O(|\ms{\pig{\pc}}|)$ time.
Aside from the $O(|X||\pc|^2)$ term, the bottleneck of the algorithm is the first nested for loop and calculating $\pmc{\pig{\pc}}$, so the entire algorithm runs in $\complexity{(r|\pc|)}$ time.
\qed\end{proof}

\textit{Proof of Theorem \ref{result_pp}.}
By Lemma \ref{lem_pp_reduction}, it suffices to compute the $\ind{\pc}-$minimum fill of $\pig{\pc}$.
This takes $\complexity{(r|\pc|)}$ time by Theorem \ref{thm_algorithm}.
\qed

\textit{Proof of Theorem \ref{result_binary_cr}.}
By Theorem \ref{thm_binary_cr_reduction}, it suffices to compute the $\fwt-$minimum / $\ind{\pc}$-minimum fill of $\pig{\pc}$.
Each character has only two states, so this takes $\complexity{|\pc|}$ time by Theorem \ref{thm_algorithm}.
\qed

\textit{Proof of Theorem \ref{result_upp}.}
By Theorem \ref{thm_upp_reduction}, it suffices to determine if $\minms{\ind{\pc}}{\pig{\pc}}$ is a pairwise-parallel set of minimal separators.
Computing $\minms{\ind{\pc}}{\pig{\pc}}$ takes $\complexity{(r|\pc|)}$ time by Theorem \ref{thm_algorithm}.
To determine if $S$ and $S' \in \minms{\ind{\pc}}{\pig{\pc}}$ are parallel, we compute the connected components of $G - S$ in linear time, and then count the number of connected components that have a vertex from $S'$.
$S$ and $S'$ are parallel if and only if this count is one.
This takes at most $O((r|\pc|)^2|\minms{\ind{\pc}}{\pig{\pc}}|^2)$ time, and $\minms{\ind{\pc}}{\pig{\pc}} \subseteq \ms{\pig{\pc}}$ giving a total of $\complexity{(r|\pc|)}$ time.
\qed

\section{Discussion}
An immediate question is whether or not Theorem \ref{result_binary_cr} can be extended to the case where $r$ is unbounded.
This does not seem possible to do for the following reason.
Let $S$ be a minimal separator that is a clique in a minimal triangulation of $\pig{\pc}$ that is an optimal solution to the maximum compatibility problem, and $C_1, C_2, \ldots, C_k$ be the connected components of $\pig{\pc} - S$.
Then the minimal triangulations of $R(S,C_i)$ for $1 \leq i \leq k$ are dependent with respect to a fill weight $\fwt$, unlike the two-state case, as illustrated in Figure \ref{fig_separator}.
It is possible to construct similar examples for unweighted characters.
Hence any sort of separator-based approach for a given optimization function seems to require a decomposition property similar to that of Lemma \ref{lem_decompose}.

\begin{figure}[h]
\centering
\subfigure[$H$]{
   \scalebox{.65}{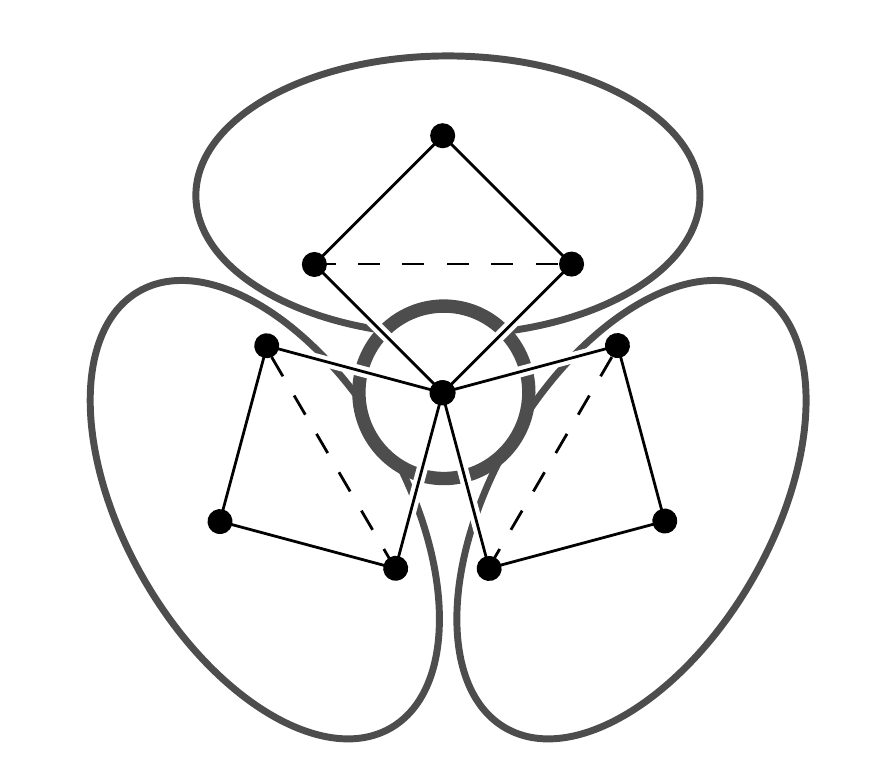}
 }
\subfigure[$H^*$]{
   \scalebox{.65}{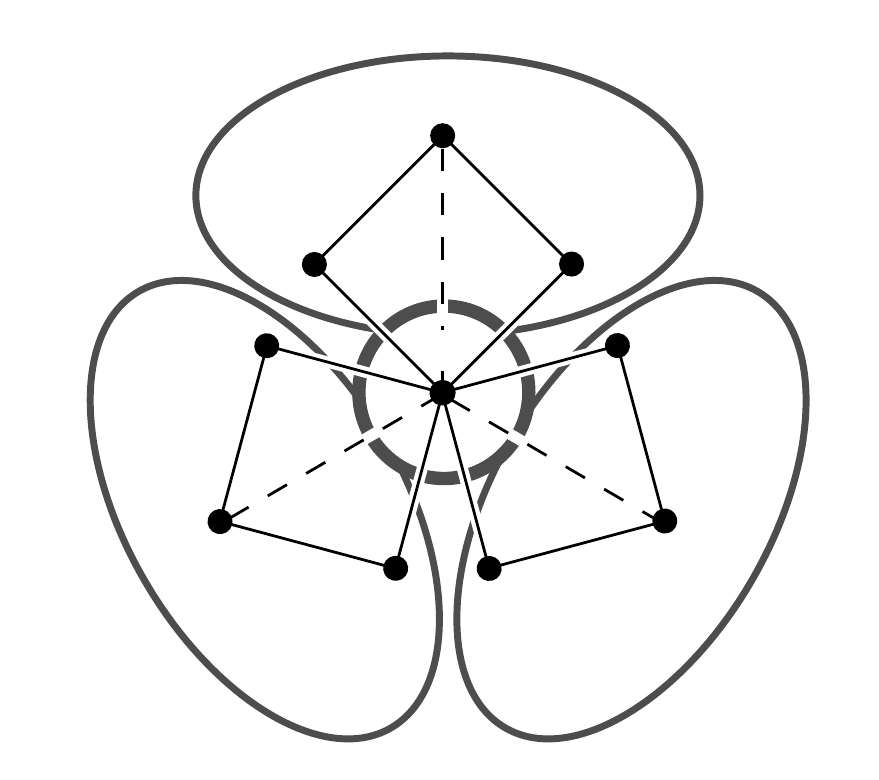}
 }
\caption{
Two triangulations of $\pig{\pc}$ with $\pc = \{abcdef|gh|ij|kl, ag|bh, ci|dj, ek|fl\}$ (as before, $A = abcdef$). 
The connected components have been drawn when $S = \{abcdef\}$ is removed from the graph.
When the characters are weighted by $\wt(abcdef|gh|ij|kl) = 2$ and $\wt(ag|bh) = \wt(ci|dj) = \wt(ek|fl) = 1$, the $\fwt-$minimum fill of the entire graph is given by the fill edges $(abcdef,gh)$, $(abcdef,ij)$, and $(abcdef,kl)$, resulting in $H^*$.
However, for each full block $R(S,C_i)$, the $\fwt-$minimum fill of that block is given adding one of the following fill edges: either $(ag,bh)$, $(ci,dj)$, or $(ek,fl)$, which leads to the suboptimal triangulation $H$.}\label{fig_separator}
\end{figure}

\section{Acknowledgements}
This research was partially supported by NSF grants IIS-0803564 and CCF-1017580.

\bibliographystyle{plain}
\bibliography{arxiv_preprint_1}

\end{document}

%% file: pp.pdf_tex
%% Creator: Inkscape 0.48.2, www.inkscape.org
%% Accompanies image file 'pp.pdf' (pdf, eps, ps)
\begingroup%
  \makeatletter%
  \providecommand\color[2][]{%
    \errmessage{(Inkscape) Color is used for the text in Inkscape, but the package 'color.sty' is not loaded}%
    \renewcommand\color[2][]{}%
  }%
  \providecommand\transparent[1]{%
    \errmessage{(Inkscape) Transparency is used (non-zero) for the text in Inkscape, but the package 'transparent.sty' is not loaded}%
    \renewcommand\transparent[1]{}%
  }%
  \providecommand\rotatebox[2]{#2}%
  \ifx\svgwidth\undefined%
    \setlength{\unitlength}{374.175bp}%
    \ifx\svgscale\undefined%
      \relax%
    \else%
      \setlength{\unitlength}{\unitlength * \real{\svgscale}}%
    \fi%
  \else%
    \setlength{\unitlength}{\svgwidth}%
  \fi%
  \global\let\svgwidth\undefined%
  \global\let\svgscale\undefined%
  \makeatother%
  \begin{picture}(1,0.58628984)%
    \put(0,0){\includegraphics[width=\unitlength]{pp.pdf}}%
    \put(0.01791063,0.22507877){\color[rgb]{0,0,0}\makebox(0,0)[lb]{\smash{\scalebox{1.5}{$g$}}}}%
    \put(0.11898103,0.02795565){\color[rgb]{0,0,0}\makebox(0,0)[lb]{\smash{\scalebox{1.5}{$h$}}}}%
    \put(0.11915123,0.53904216){\color[rgb]{0,0,0}\makebox(0,0)[lb]{\smash{\scalebox{1.5}{$b$}}}}%
    \put(0.01791062,0.34980396){\color[rgb]{0,0,0}\makebox(0,0)[lb]{\smash{\scalebox{1.5}{$a$}}}}%
    \put(0.32398907,0.01648666){\color[rgb]{0,0,0}\makebox(0,0)[lb]{\smash{\scalebox{1.5}{$i$}}}}%
    \put(0.52182902,0.01648666){\color[rgb]{0,0,0}\makebox(0,0)[lb]{\smash{\scalebox{1.5}{$j$}}}}%
    \put(0.6580232,0.41431699){\color[rgb]{0,0,0}\makebox(0,0)[lb]{\smash{\scalebox{1.5}{$d$}}}}%
    \put(0.96625203,0.22507877){\color[rgb]{0,0,0}\makebox(0,0)[lb]{\smash{\scalebox{1.5}{$l$}}}}%
    \put(0.86804895,0.02795564){\color[rgb]{0,0,0}\makebox(0,0)[lb]{\smash{\scalebox{1.5}{$k$}}}}%
    \put(0.86804894,0.53904213){\color[rgb]{0,0,0}\makebox(0,0)[lb]{\smash{\scalebox{1.5}{$e$}}}}%
    \put(0.96768568,0.34980396){\color[rgb]{0,0,0}\makebox(0,0)[lb]{\smash{\scalebox{1.5}{$f$}}}}%
    \put(0.45946644,0.41431697){\color[rgb]{0,0,0}\makebox(0,0)[lb]{\smash{\scalebox{1.5}{$c$}}}}%
  \end{picture}%
\endgroup%

%% file: incompatible.pdf_tex
%% Creator: Inkscape 0.48.2, www.inkscape.org
%% Accompanies image file 'incompatible.pdf' (pdf, eps, ps)
\begingroup%
  \makeatletter%
  \providecommand\color[2][]{%
    \errmessage{(Inkscape) Color is used for the text in Inkscape, but the package 'color.sty' is not loaded}%
    \renewcommand\color[2][]{}%
  }%
  \providecommand\transparent[1]{%
    \errmessage{(Inkscape) Transparency is used (non-zero) for the text in Inkscape, but the package 'transparent.sty' is not loaded}%
    \renewcommand\transparent[1]{}%
  }%
  \providecommand\rotatebox[2]{#2}%
  \ifx\svgwidth\undefined%
    \setlength{\unitlength}{153.15bp}%
    \ifx\svgscale\undefined%
      \relax%
    \else%
      \setlength{\unitlength}{\unitlength * \real{\svgscale}}%
    \fi%
  \else%
    \setlength{\unitlength}{\svgwidth}%
  \fi%
  \global\let\svgwidth\undefined%
  \global\let\svgscale\undefined%
  \makeatother%
  \begin{picture}(1,1.42539993)%
    \put(0,0){\includegraphics[width=\unitlength]{incompatible.pdf}}%
    \put(0.47137464,0.69980573){\color[rgb]{0,0,0}\makebox(0,0)[lb]{\smash{\scalebox{1.5}{$A$}}}}%
    \put(0.55177706,0.18728439){\color[rgb]{0,0,0}\makebox(0,0)[lb]{\smash{\scalebox{1.5}{$fl$}}}}%
    \put(0.37852483,0.18852285){\color[rgb]{0,0,0}\makebox(0,0)[lb]{\smash{\scalebox{1.5}{$bh$}}}}%
    \put(0.47183915,1.18577675){\color[rgb]{0,0,0}\makebox(0,0)[lb]{\smash{\scalebox{1.5}{$ij$}}}}%
    \put(0.87723303,0.70138263){\color[rgb]{0,0,0}\makebox(0,0)[lb]{\smash{\scalebox{1.5}{$ek$}}}}%
    \put(0.04488063,0.70177103){\color[rgb]{0,0,0}\makebox(0,0)[lb]{\smash{\scalebox{1.5}{$ag$}}}}%
    \put(0.88370849,0.27599835){\color[rgb]{0,0,0}\makebox(0,0)[lb]{\smash{\scalebox{1.5}{$kl$}}}}%
    \put(0.79118071,0.8551796){\color[rgb]{0,0,0}\makebox(0,0)[lb]{\smash{\scalebox{1.5}{$dj$}}}}%
    \put(0.14949105,0.85496711){\color[rgb]{0,0,0}\makebox(0,0)[lb]{\smash{\scalebox{1.5}{$ci$}}}}%
    \put(0.04795966,0.28682885){\color[rgb]{0,0,0}\makebox(0,0)[lb]{\smash{\scalebox{1.5}{$gh$}}}}%
  \end{picture}%
\endgroup%

%% file: separator.pdf_tex
%% Creator: Inkscape 0.48.2, www.inkscape.org
%% Accompanies image file 'separator.pdf' (pdf, eps, ps)
\begingroup%
  \makeatletter%
  \providecommand\color[2][]{%
    \errmessage{(Inkscape) Color is used for the text in Inkscape, but the package 'color.sty' is not loaded}%
    \renewcommand\color[2][]{}%
  }%
  \providecommand\transparent[1]{%
    \errmessage{(Inkscape) Transparency is used (non-zero) for the text in Inkscape, but the package 'transparent.sty' is not loaded}%
    \renewcommand\transparent[1]{}%
  }%
  \providecommand\rotatebox[2]{#2}%
  \ifx\svgwidth\undefined%
    \setlength{\unitlength}{253.05bp}%
    \ifx\svgscale\undefined%
      \relax%
    \else%
      \setlength{\unitlength}{\unitlength * \real{\svgscale}}%
    \fi%
  \else%
    \setlength{\unitlength}{\svgwidth}%
  \fi%
  \global\let\svgwidth\undefined%
  \global\let\svgscale\undefined%
  \makeatother%
  \begin{picture}(1,0.88717645)%
    \put(0,0){\includegraphics[width=\unitlength]{separator.pdf}}%
    \put(0.48267546,0.47619144){\color[rgb]{0,0,0}\makebox(0,0)[lb]{\smash{\scalebox{1.5}{$A$}}}}%
    \put(0.53133632,0.16600514){\color[rgb]{0,0,0}\makebox(0,0)[lb]{\smash{\scalebox{1.5}{$fl$}}}}%
    \put(0.42648124,0.16675468){\color[rgb]{0,0,0}\makebox(0,0)[lb]{\smash{\scalebox{1.5}{$bh$}}}}%
    \put(0.4829566,0.77030906){\color[rgb]{0,0,0}\makebox(0,0)[lb]{\smash{\scalebox{1.5}{$ij$}}}}%
    \put(0.7283076,0.47714581){\color[rgb]{0,0,0}\makebox(0,0)[lb]{\smash{\scalebox{1.5}{$ek$}}}}%
    \put(0.22455431,0.47738087){\color[rgb]{0,0,0}\makebox(0,0)[lb]{\smash{\scalebox{1.5}{$ag$}}}}%
    \put(0.73222666,0.21969628){\color[rgb]{0,0,0}\makebox(0,0)[lb]{\smash{\scalebox{1.5}{$kl$}}}}%
    \put(0.67622733,0.57022625){\color[rgb]{0,0,0}\makebox(0,0)[lb]{\smash{\scalebox{1.5}{$dj$}}}}%
    \put(0.28786625,0.57009765){\color[rgb]{0,0,0}\makebox(0,0)[lb]{\smash{\scalebox{1.5}{$ci$}}}}%
    \put(0.22641779,0.22625108){\color[rgb]{0,0,0}\makebox(0,0)[lb]{\smash{\scalebox{1.5}{$gh$}}}}%
  \end{picture}%
\endgroup%

%% file: separator2.pdf_tex
%% Creator: Inkscape 0.48.2, www.inkscape.org
%% Accompanies image file 'separator2.pdf' (pdf, eps, ps)
\begingroup%
  \makeatletter%
  \providecommand\color[2][]{%
    \errmessage{(Inkscape) Color is used for the text in Inkscape, but the package 'color.sty' is not loaded}%
    \renewcommand\color[2][]{}%
  }%
  \providecommand\transparent[1]{%
    \errmessage{(Inkscape) Transparency is used (non-zero) for the text in Inkscape, but the package 'transparent.sty' is not loaded}%
    \renewcommand\transparent[1]{}%
  }%
  \providecommand\rotatebox[2]{#2}%
  \ifx\svgwidth\undefined%
    \setlength{\unitlength}{253.05bp}%
    \ifx\svgscale\undefined%
      \relax%
    \else%
      \setlength{\unitlength}{\unitlength * \real{\svgscale}}%
    \fi%
  \else%
    \setlength{\unitlength}{\svgwidth}%
  \fi%
  \global\let\svgwidth\undefined%
  \global\let\svgscale\undefined%
  \makeatother%
  \begin{picture}(1,0.88717645)%
    \put(0,0){\includegraphics[width=\unitlength]{separator2.pdf}}%
    \put(0.53133632,0.16600515){\color[rgb]{0,0,0}\makebox(0,0)[lb]{\smash{\scalebox{1.5}{$fl$}}}}%
    \put(0.42648124,0.16675469){\color[rgb]{0,0,0}\makebox(0,0)[lb]{\smash{\scalebox{1.5}{$bh$}}}}%
    \put(0.4829566,0.77030906){\color[rgb]{0,0,0}\makebox(0,0)[lb]{\smash{\scalebox{1.5}{$ij$}}}}%
    \put(0.7283076,0.4771458){\color[rgb]{0,0,0}\makebox(0,0)[lb]{\smash{\scalebox{1.5}{$ek$}}}}%
    \put(0.22455431,0.47738086){\color[rgb]{0,0,0}\makebox(0,0)[lb]{\smash{\scalebox{1.5}{$ag$}}}}%
    \put(0.73222666,0.21969629){\color[rgb]{0,0,0}\makebox(0,0)[lb]{\smash{\scalebox{1.5}{$kl$}}}}%
    \put(0.67622733,0.57022625){\color[rgb]{0,0,0}\makebox(0,0)[lb]{\smash{\scalebox{1.5}{$dj$}}}}%
    \put(0.28786625,0.57009765){\color[rgb]{0,0,0}\makebox(0,0)[lb]{\smash{\scalebox{1.5}{$ci$}}}}%
    \put(0.22641779,0.22625109){\color[rgb]{0,0,0}\makebox(0,0)[lb]{\smash{\scalebox{1.5}{$gh$}}}}%
    \put(0.48267546,0.47619146){\color[rgb]{0,0,0}\makebox(0,0)[lb]{\smash{\scalebox{1.5}{$A$}}}}%
  \end{picture}%
\endgroup%